\documentclass[dvipsnames]{article}
\usepackage[utf8]{inputenc}
\usepackage{amsmath,amsthm,amsfonts}

\usepackage{url}
\usepackage{wrapfig}
\usepackage{multirow}

\usepackage{hyperref}
\hypersetup{colorlinks,linkcolor=black,filecolor=black,urlcolor=black,citecolor=black}

\usepackage{pstricks}
\newcommand{\violet}{\color{violet}}

\theoremstyle{plain}
\newtheorem{theorem}{Theorem}[section]
\newtheorem{lemma}[theorem]{Lemma}

\theoremstyle{definition}
\newtheorem{definition}[theorem]{Definition}
\newtheorem{remark}[theorem]{Remark}

\usepackage{cite}

\usepackage{pdftricks}
\begin{psinputs}
\newcommand{\violet}{\color{violet}}
\newcommand{\cH}{\mathcal{H}}
\newcommand{\cW}{\mathcal{W}}
\newcommand{\cS}{\mathcal{S}}
\newcommand{\Ws}{\cW_{\text{succ}}}
\newcommand{\Wf}{\cW_{\text{fail}}}
\newcommand{\Ss}{\cS_{\text{succ}}}
\newcommand{\Sf}{\cS_{\text{fail}}}
\usepackage{pstricks}
\usepackage{pstricks-add}
\usepackage{color}
\usepackage{pstcol}
\usepackage{pst-plot}
\usepackage{pst-tree}
\usepackage{pst-eps}
\usepackage{multido}
\usepackage{pst-node}
\usepackage{pst-eps}
\end{psinputs}
\newcommand{\brac}[1]{\left(#1\right)}
\newcommand{\OO}{\mathcal{O}}
\newcommand{\cH}{\mathcal{H}}
\newcommand{\cI}{\mathcal{I}}
\newcommand{\cE}{\mathcal{E}}

\newcommand{\cW}{\mathcal{W}}
\newcommand{\cS}{\mathcal{S}}
\newcommand{\Ws}{\cW_{\text{succ}}}
\newcommand{\Wf}{\cW_{\text{fail}}}
\newcommand{\Ss}{\cS_{\text{succ}}}
\newcommand{\Sf}{\cS_{\text{fail}}}

\newcommand{\Pos}{\mathrm{Poisson}}
\newcommand{\Gam}{\mathrm{Gamma}}

\newcommand{\T}{\mathrm{T}}
\newcommand{\F}{\mathrm{F}}

\title{The paradoxical nature of easily improvable evidence}
\author{Maria Chikina\footnote{Department of Computational and Systems Biology, University of Pittsburgh, Pittsburgh, PA 15260}, Wesley Pegden\footnote{Department of Mathematical Sciences, Carnegie Mellon University, Pittsburgh, PA 15213.  Correspondence to wes@math.cmu.edu}}
\begin{document}

\maketitle

\begin{abstract}
    Established frameworks to understand problems with reproducibility in science begin with the relationship between our understanding of the prior probability of a claim and the statistical certainty that should be demanded of it, and explore the ways in which independent investigations, biases in study design and publication bias interact with these considerations.
    
    We propose a complementary perspective; namely, that to improve reproducibility in science, our interpretation of the persuasiveness of evidence (e.g., statistical significance thresholds) should be responsive to our understanding of the effort that would be required to improve that evidence.  We will quantify this notion in some formal settings.  Indeed, we will demonstrate that even simplistic models of evidence publication can exhibit an \emph{improvable evidence paradox}, where the publication of easily improvable evidence in favor of a claim can best seen as evidence the claim is false. 
\end{abstract}
\section{Introduction}
In the face of widespread concern about the reliability of published scientific research, large-scale replication studies have found low reproducibility across a variety of scientific fields \cite{manylabs2,economics,open2015estimating,bayer,amgen}.  Some causes such as small sample size, small effect sizes, and  $p$-hacking can be broadly characterized as lack of statistical rigour.   Another important contributor is publication bias, the practice of only publishing positive results. Publication bias can obscure the fact that multiple attempts to demonstrate a hypothesis were made and thus inflate our confidence in positive findings. 
Practices such as accepting papers based on the experimental design alone (irrespective of the results or even prior to their obtainment) have gained momentum but have not been widely adopted \cite{SterneDavey,manifesto}. Thus, publication bias remains a fundamental property of the scientific process in most fields \cite{glenjohn}.   


Established frameworks to understand the predictive value of positive research findings  leverage the relationship between the prior probability of the hypothesis in question, the bias which may be at work in investigations, and the extent to which the work of multiple independent investigators may be shielded from view by publication bias \cite{whymost}.   These frameworks demonstrate that in the face of these factors, the observation of successful studies of a hypothesis can sometimes have only a moderate effect on the posterior probability that a hypothesis $H$ is true.  In particular, when the prior probability of $H$ is low, $H$ can remain unlikely to be true even after successful studies are observed, as supported by the high failure rates of replication attempts.  Apart from approaches to address publication bias directly, proposals to address these challenges focus on improving the statistical power of studies, emphasizing the importance of effect size and study design above the arbitrary division of results into ``significant'' and ``non-significant'' based on moderate significance thresholds, and reducing the potential for design and investigator bias \cite{manifesto,button2013power,ioannidis2017power,SterneDavey,redefine}.

We propose a complementary idea, which is that the false-positive and false-negative rates (Type I and Type II errors) of studies should be evaluated not just in the context of the prior probability of a hypothesis
and the level of confidence we desire in its truth, 
but also  the extent to which improving the statistical power of the the study protocol would be difficult or costly.  The reason for this is not merely pragmatic (though, certainly, it can be sensible to pursue better confidence in findings when that confidence is cheap).  Instead, we will show that in a framework where the studies  conducted to support a hypothesis do not all have identical Type I and Type II error rates but are  instead heterogeneous, the posterior confidence in the hypothesis depends not only on the error rates of the successful studies we observe, but also on the rate at which we would expect stronger evidence for the hypothesis to be generated.  
In particular, we will prove in a simple and rigorous framework that when our prior on the interest-level in a scientific hypothesis is sufficiently weak, we can encounter an \emph{improvable evidence paradox}: \textbf{The more easily-improvable evidence in favor of the hypothesis we observe, the more we should become convinced that the hypothesis is false.}

This phenomenon provides a new lens through which to understand the limitations of meta-analyses which infer ever-increasing confidence in an effect with each additional observation of a small, barely significant study.  Specific examples where such meta-analyses fail to replicate in subsequent large-scale trials are given in \cite{daveyeditorial,funnel}.  In particular, our results suggest that meta-analysis approaches which integrate individual pieces of weak evidence without regard to their improvability should have a tendency to overstate confidence in effects.

\section{A motivating illustration}
In this section we give an informal motivating example, to give some intuition for why it is reasonable to suspect that easily improvable evidence can sometimes have paradoxical effects.

We consider the case of a village full of curious investigators, living on the outskirts of a forest teeming with unknown animal life.  From stories and legends, the villagers have a list of dozens of hypothetical animal inhabitants; for each item on the list, the villagers wonder whether an animal with that description truly lives in the forest.  

Suppose the villagers can collect two types of evidence---weak and strong---for the existence of an animal from the list in the forest.  Weak evidence might consist of a blurry photograph that seems to depict the animal, taken with a low-tech camera, while strong evidence might consist of clear, high-resolution photographs or video.

In such a situation, we might expect that the villagers' imagination is captured more by some of the hypothetical animals than by others.  For those they are most excited about, many villagers go in search of evidence the animal exists, while the most boring animals on their list go completely uninvestigated.

With this setup in mind, suppose now that we visit this village, and visit a museum that displays all the evidence they've collected about these animals.  We'll see that for some animals, villagers will have collected lots of evidence, both strong and weak, that the animal exists.  These hypothetical animals have two properties; they are likely real \emph{and} the villagers are very interested them, having investigated their existence extensively.

There may be hypothetical animals for which the museum contains no evidence at all.  These animals may not exist, or, alternatively, villagers may simply not have been interested enough in them to try to collect evidence of their existence.  If half of all the animals on the villagers' list actually exist in the forest, it may well be the case that nearly half of the animals completely unrepresented in the museum actually exist as well.

On the other hand, when we encounter animals for which the museum contains plenty of weak evidence, but no strong evidence, we are in a different situation.  We know the villagers are interested enough in these hypothetical animals to have collected a lot of evidence of their existence.  But they have not managed to collect any strong evidence; only weak evidence.  The most parsimonious explanation is that these animals are unlikely to actually inhabit the forest---in particular, they are less likely to inhabit the forest than the animals for which the museum contains no evidence at all!  Every  additional blurry photograph we see of ``Bigfoot'' isn't just weak evidence of its existence; it's yet another reminder that there is considerable interest in this specimen, and yet no solid documentation. 

We will study this phenomenon in a simple and rigorous model of evidence generation.  We will see that paradoxical effects of easily improvable evidence do not depend on our prior probability that a hypothesis is true (e.g., the background rate at which animals on the villagers' list actually exist) but really on the improvability of the evidence, as captured by the rate at which evidence with more favorable Type I and Type II error rates should be sought.  We will see that the paradoxical nature of improvable evidence arises even when we assume that evidence generation is unbiased, and even if we assume that decisions to seek strong and weak evidence are made independently given the interest level in the hypothesis.  We will also see that even the first observation of easily improvable evidence can be paradoxical; in particular, we will prove in a natural model that for sufficiently overdispersed priors on the interest-level in a hypothesis, every additional observation of weak evidence reduces the conditional probability that the hypothesis in question is true.


\section{Positive predictive value of evidence}

Before introducing our setup, let us recall that the positive predictive value (PPV) of the observation of evidence---denoted as event $\cE$---in favor of the truth of a hypothesis---an event $\cH$---is simply the conditional probability $\Pr(\cH\mid \cE)$, which we can write as
\begin{equation}\label{eq.oneW}
\Pr(\cH\mid \cE)=\Pr(\cE\mid \cH) \frac{\Pr(\cH)}{\Pr(\cE)}
=\frac{\Pr(\cH)}{\Pr(\cH)+(1-\Pr(\cH))/C_{\cH,\cE}}
\end{equation}
for the likelihood ratio
\begin{equation}\label{eq.CHE}
C_{\cH,\cE}:=\frac{\Pr(\cE\mid \cH)}{\Pr(\cE\mid \neg \cH)}.
\end{equation}
In particular, in a standard framework where $\cE$ is the success of a single study with false-positive and false-negative rates $\alpha$ and $(1-\gamma)$---alternatively with \emph{power} $\gamma$---we have
$
C_{\cH,\cE}={\gamma}/{\alpha}.
$
For example, for a study that is underpowered by conventional standards, with Type I error $\alpha=.05$ and power $\gamma=.2$, we have $C_{\cH,\cE}=4$ and  $\Pr(\cH\mid \cE)$ would be 50\% if the prior $\Pr(\cH)$ is $20\%$, and approximately $31\%$ for a prior of $10\%$.  

When publication bias shields from view all but positive studies for a hypothesis, a simple framework to understand the effect on PPV is to consider events %
%
%
%
$    \OO_{\geq j}=\{\text{Observe at least $j$ successes}\}
$
(with the lens of publication bias concealing any negative results).
If $n$ studies were attempted, we would then have
\begin{equation}\label{eq.lambdan}
C_{\cH,\OO_{\geq 1}}=\frac{1-(1-\gamma)^n}{1-(1-\alpha)^n}\quad 
\end{equation}
For example, considering again the case $\gamma=.2$, Type I error $\alpha=.05$, where now we expect 5 independent such studies have been done to test the hypothesis, we would have $\Pr(\cH\mid \OO_{\geq 1})\approx 43\%$ for a prior of $\Pr(\cH)=20\%$ and $\approx 24\%$ for a prior of $10\%$.

While we see that role of publication bias can make PPV's weak, this classical framework for understanding the predictive value of scientific evidence does at least give a simple condition which ensures that each observation of successful evidence can only increase the conditional probability of $\cH$; in particular, we have in these cases that $C_{\cH,\cE}>1$ and thus that  $\Pr(\cH\mid \cE)$ is always greater than the unconditioned prior $\Pr(\cH)$, so long as the inequality $\gamma>\alpha$ holds---in particular, whenever a study is more likely to succeed when the hypothesis is true than when it is false.  If we define the events $\OO_j=\{\text{Observe exactly $j$ successes}\}$, it is easy to verify that $C_{\cH,\OO_j}$ is monotone increasing in $j$ so long as $\gamma>\alpha$.

A key simplifying feature of this framework as presented above is that all studies are assumed have identical designs, or at least identical false-positive and false-negative rates.   When we generalize this framework to the setting where evidence of different strength can be sought for a hypothesis, the impact of positive results on our credence for the truth of $\cH$ no longer depends simply on the ratio $\gamma/\alpha$. Instead, the conditional probability of H also depends on how easy it should be to \emph{improve} on the evidence we've observed.
In particular, we demonstrate in a simple and rigorous framework that in the presence of heterogeneity of the strength of evidence that can be generated for a hypothesis, observing more evidence can sometimes actually decrease the conditional probability $\cH$ holds, so long as better evidence should be easily obtainable.

\bigskip Roughly speaking, what's going is this: When a variety of studies can be done to test a scientific hypothesis and publication is biased towards only successful studies, observing successful outcomes imparts not only information about the truth of the hypothesis, but also the level of scientific interest in the hypothesis as well.  Observing publication of easily improvable evidence thus influences our sense of the truth of $H$ in counteracting ways:
\begin{enumerate}
    \item On the one hand, any attempts to generate the evidence would be more likely to succeed if $H$ is true;
    \item On the other, generating the evidence was more likely to be attempted if there is considerable interest in the hypothesis $H$, in which case the fact that we haven't observed stronger evidence of $H$ is also informative.
\end{enumerate}
What we show in this paper is that even in simple models of observational frameworks, the second factor can sometimes dominate, for observations of easily improvable evidence.

In practice, the improvability of evidence varies greatly across research questions.  Studies on humans requiring invasive measurements or treatments present clear costs with increasing sample size; at the other end of the spectrum, the marginal cost of increasing the study population for survey-based studies (perhaps conducted entirely online, for example) in psychology or economics is much lower.  Most scientific research falls between these two extremes, but as Type I and Type II error rates decay exponentially with sample size, the improvability of statistical reliability is often underestimated \cite{redefine}.

\section{Formal setting and the main result}

We will define a simple formal model of an \emph{observational framework}, governing the generation and observation of weak and strong evidence.  Within such a model, attempts to generate strong and weak evidence in favor of a hypothesis will be generated; we are analyzing such a system from the standpoint of an observer who only learns of successful outcomes (i.e., through the lens of publication bias).  Within such a model, we can define
\begin{equation}\label{eq.Ojk}
    \OO_{j,k}=\text{Observe exactly $j$ weak successes and $k$ strong successes}.\\
\end{equation}
We are then interested in understanding, based on observations of successful experiments, what we should conclude about the probability of the truth of the hypothesis.
\begin{definition}
An observation framework exhibits an \emph{improvable evidence paradox} whenever, for all $j,k\geq 0$,
\begin{equation}\label{eq.oparadox}
    \Pr(\cH\mid \OO_{j+1,k})<\Pr(\cH\mid \OO_{j,k}).
\end{equation}
\end{definition}

In particular, we have an improvable evidence paradox whenever an increase in observed successes of weak experiments decreases the conditional probability that the hypothesis $H$ is true.

\bigskip

For our model, we will consider the example of a true-or-false hypothesis $H$ which may be tested in experiments, at a rate that will be influenced by a \emph{level of interest} $\cI>0$.  We consider a simple model that is already sufficient to observe the improvable evidence paradox, in which:
\begin{itemize}
    \item Two experimental outcomes are possible: success or failure.
    \item $H$ can be tested in experiments of two types: strong and weak
        \begin{itemize}
     \item Success of each weak experiment on false and true hypotheses has probability $\alpha_w$ and $\gamma_w$, respectively (Type I and Type II error),
     \item Success of each strong experiment on false and true hypotheses has probability $\alpha_S$ and $\gamma_S$, respectively,
     \end{itemize}
\end{itemize}

The probability space for our main result is constructed in the following way, with each step performed independently of previous steps except where specified:
\begin{itemize}
    \item An interest level $\cI$ for the hypothesis $H$ is chosen from a specified prior distribution,
     \item The number of weak and strong experiments that will be conducted is chosen independently from Poisson distributions of rates $c_w\cI$ and $\cI$, respectively for some constant $c_w>0$,
     \item Whether the event $\cH$ (``$H$ is true'') occurs is chosen by a coin flip with a probability $\Pr(\cH)$,
     \item If $H$ is true, the success of each weak and strong experiment is determined by the outcomes of independent coin flips with success probabilities $\gamma_w$ and $\gamma_S$, respectively.  If $H$ is not true, the outcomes are decided by independent coin flips with success probabilities $\alpha_w$ and $\alpha_S$, respectively.
\end{itemize}
\noindent We will assume 
\begin{equation}\label{eq:stronger}
\frac{\gamma_S}{\alpha_S}>\frac{\gamma_w}{\alpha_w}>1;
\end{equation}
in particular, given two types of studies or evidence, the first inequality is what determines which we consider ``strong'' vs ``weak''.

\bigskip 

We call this probability space determining $\cH$, $\cI$, and the number of weak and strong experiments that are conducted and that succeed an \emph{interest-$\cI$ observation framework}.

\begin{theorem}\label{t.main}
Any interest-$\cI$ observation framework for which the interest level $\cI$ is gamma-distributed with rate $\beta$ satisfying
\begin{equation}\label{eq.main}
    \beta<\frac{\gamma_S\alpha_w-\gamma_w\alpha_S}{\gamma_w-\alpha_w},
\end{equation}
exhibits an improvable evidence paradox.
\end{theorem}
\noindent The proof of the Theorem is given in Appendix \ref{s.proof}.

Recall that the rate $\beta$ of a gamma distribution is the ratio of its mean to its variance; thus, an improvable evidence paradox occurs when the variance in our prior for a gamma-distributed interest level of a hypothesis is sufficiently large---as captured by \eqref{eq.main}---relative to its mean.  (Note that condition \eqref{eq:stronger} implies that the righthandside of \eqref{eq.main} is greater than 0.) Informally: When we are sufficiently uncertain about the interest-level in a hypothesis among investigators, funding agencies, and journals, an increase in our estimate of the true interest level in the hypothesis (and thus, an increased expectation that stronger evidence would be published for the hypothesis, if it is indeed true) can be, in a Bayesian sense, the most important impact of observing the publication of easily-improvable evidence with poor Type I or Type II error rates.

For the sake of concreteness, let us return to the example of a hypothesis that may be tested in underpowered studies, with false-positive rate $\alpha=.05$ and power $\gamma=.2$.  As discussed earlier, success of such an underpowered study may have just a modest positive effect on the conditional probability that $\cH$ is true, particularly if several such studies may have been carried out, with publication bias only revealing successes.  But now suppose that strong studies, say with power $\gamma=.9$ and $\alpha=.01$ might also be conducted (at a rate related to the rate of underpowered studies, through the interest-level in the hypothesis).  In this case the improvable evidence paradox occurs whenever the inverse rate of the prior for the interest level is at least
\[
\frac{\gamma_w-\alpha_w}{\gamma_S\alpha_w-\gamma_w\alpha_S}=\frac{.2-.05}{.9\times .05-.4\times .01}\approx 3.7.
\]
In particular, if the variance of our gamma-distributed prior for the interest level for hypotheses in the field of $H$ at is least 3.7 times the mean (for example, this is true in the special case where the prior is exponentially distributed with mean greater than    3.7), then not only is each successful weak experiment we observe not very informative for $H$; each such observation actually makes the truth of the hypothesis less likely.

\bigskip
We note that paradoxical behavior of improvable evidence does not depend on unique properties of the choice of the gamma distribution as our prior for $\cI$. For example, let us say that an observation framework satisfies a  improvable evidence paradox \emph{up to $K$} if \eqref{eq.oparadox} holds for all $j+k\leq K$.  In Section \ref{s.proof} we will also show the following:
\begin{theorem}\label{t.uniform}
A $\cI$ observation framework for which $\cI$ is uniformly distributed in $[0,C]$ exhibits an improvable evidence paradox up to $K$ if $C$ is sufficiently large given $K$ and $\gamma_w, \gamma_S, \alpha_w,\alpha_S$, and $c_w$.
\end{theorem}

\begin{remark}
Throughout this paper we work with Type I and II error rates that are separated depending on whether a hypothesis $H$ is ``true'' or ``false''.  Note that for many hypotheses (e.g., concerning an effect of some medication), this implicitly corresponds to an assumption that a hypothesis is true with some minimum effect size, or else not true at all (i.e., the null hypothesis holds).  Alternatively, one could, for example, model effect sizes as either 0 or normally distributed, and consider the positive predictive value of observations to be the probability, conditioned on those observations, that we are in the second situation.  We eschew this more complicated model here just for the sake of simplicity.  (In practice this extra complication would only increase the paradoxical nature of improvable evidence, by giving an extra way that weak and strong successes are correlated.)
\end{remark}
\section{A more general setting}

\begin{figure}[!p]
\begin{center}
\begin{pdfpic}
    \psset{unit=.7cm}
    \begin{pspicture}(0,0)(10,10)
    \psframe[framearc=.2](0,0)(10,10)
    \psline[linewidth=1.2pt](5,0)(5,10)
    \rput(4,9){\Large $\cH$}
    \rput(6,9){\Large $\cH^C$}

    \psframe[linecolor=red,fillstyle=hlines,hatchwidth=.5pt,hatchcolor=red,linewidth=1.2pt,framearc=.3](1.5,3.25)(5.15,7)
    \psframe[linecolor=red,fillstyle=none,hatchwidth=.5pt,hatchcolor=red,linewidth=1.2pt,framearc=.3](1.5,3.25)(8.15,7)
    \rput(3.5,7){\psframebox*[framesep=0pt,framearc=.2]{\large \red $\Ss$}}
    \rput(6,7){\psframebox*[framesep=0pt,framearc=.2]{\large \red $\Sf$}}
    \rput(4.625,7.4)
    {\large \red $\cS$}}
    
    \psframe[linecolor=violet,fillstyle=vlines,hatchwidth=.5pt,hatchcolor=violet,linewidth=1.2pt,framearc=.3](2,3)(6.0,6.5)
    \psframe[linecolor=violet,fillstyle=none,hatchwidth=.5pt,hatchcolor=violet,linewidth=1.2pt,framearc=.3](2,3)(8.0,6.5)
    \rput(3.3,2.85){\psframebox*[framesep=0pt,framearc=.2]{\large \violet $\Ws$}}
    \rput(6.75,2.85){\psframebox*[framesep=0pt,framearc=.2]{\large \violet $\Wf$}}
    \rput(4.6,2.5)
    {\large \violet $\cW$}}
    
    \end{pspicture}
    \end{pdfpic}\\
    \begin{pdfpic}
        \psset{unit=.5cm}
    \begin{pspicture}(0,0)(10,10)
    \psframe[framearc=.2](0,0)(10,10)
    \psline[linewidth=1.2pt](5,0)(5,10)
    \rput(4,9){\Large $\cH$}
    \rput(6,9){\Large $\cH^C$}
    
    \rput(5.2,2.5){\large \blue ${\Ws\cap\Ss^C}$}
    \psline[linecolor=blue,fillstyle=vlines,hatchwidth=.5pt,hatchcolor=blue,linewidth=1.2pt,linearc=.3](5,3.0)(2,3.0)(2,3.25)%
    \psline[linestyle=blue,fillstyle=vlines,hatchwidth=.5pt,hatchcolor=blue,linewidth=1.2pt](2,3.25)(5,3.25)(5,3.0)
    \psline[linecolor=blue,fillstyle=none,hatchwidth=.5pt,hatchcolor=blue,linewidth=1.2pt,linearc=.3](2,3.25)(5.15,3.25)(5.15,6.5)
    \psline[linecolor=blue,fillstyle=vlines,hatchwidth=.5pt,hatchcolor=blue,linewidth=1.2pt,linearc=.3](5.15,6.5)(6.0,6.5)(6.0,3.0)(5.15,3.0)
    \psline[linecolor=blue,fillstyle=none,hatchwidth=.5pt,hatchcolor=blue,linewidth=1.2pt,linearc=.3](5.15,3.0)(5.0,3.0)
    

    
    \end{pspicture}%
    \hspace{2em}
    \begin{pspicture}(0,0)(10,10)
    \psframe[framearc=.2,fillstyle=vlines,hatchwidth=.5pt,hatchcolor=black](0,0)(10,10)
    \psline[linewidth=1.2pt](5,0)(5,10)
    \rput(4,9){\psframebox*[framearc=.2]{\Large $\cH$}}
    \rput(6,9){\psframebox*[framearc=.2]{\Large $\cH^C$}}
    
    \rput(5.0,.75){\psframebox*[framearc=.2]{\large \black ${\Ws^C\cap\Ss^C}$}}
    \psline[linecolor=black,fillstyle=solid,fillcolor=white,hatchwidth=.5pt,hatchcolor=blue,linewidth=1.2pt,linearc=.3](5.15,6.5)(6.0,6.5)(6.0,3.0)(5.15,3.0)(5.0,3.0)(2,3.0)(2,3.25)%
    \psline[linestyle=black,fillstyle=solid,fillcolor=white,hatchwidth=.5pt,hatchcolor=blue,linewidth=1.2pt,linearc=.3](2,3.25)(1.5,3.25)(1.5,7)(5.15,7)(5.15,6.5)


    
    \end{pspicture}
\end{pdfpic}
\vspace{-1cm}
\end{center}
\caption{\label{fig:WSparadox}\textbf{The improvable evidence paradox.} \emph{Top}: The probability space depicted by the rounded black square determines whether or not a hypothesis $H$ is true (event $\cH$) and whether strong or weak evidence for the hypothesis is sought for the hypothesis (events {\red $\cS=\Ss\cup \Sf$} and {\violet $\cW=\Ws\cup \Wf$}, respectively), and whether the attempts for each are successful (events $\Ss$ and $\Sf$, and $\Ws$ and $\Wf$, respectively).  When the level of interest in different hypotheses is highly variable, so that seeking strong and weak evidence is highly correlated, we can encounter the \emph{improvable evidence} paradox: observing only weak evidence can be worse than observing no evidence at all, even when, considered in isolation, the positive-predictive value of the weak evidence $\Ws$ would exceed the prior probability of $\cH$.  In this diagram, the probability of $\cH$ conditioned on {\blue $\Ws\cap \Ss^C$} (\emph{Bottom Left}) is considerably less than when conditioned on $\Ws^C\cap \Ss^C$ (\emph{Bottom Right}).}
\end{figure}

\interfootnotelinepenalty=10000

Beyond the specific context of an interest-$\cI$ observation framework as defined in the previous section, we can consider a general probability space that determines whether or not a hypothesis is true---represented by an event $\cH$---whether or not strong or weak evidence for the hypothesis is sought---represented by events $\cS$ and $\cW$,  respectively---and whether or not those attempts are successful---events $\Ss$ vs. $\Sf$ and $\Ws$ vs $\Wf$, respectively (Figure \ref{fig:WSparadox}). For simplicity in this general setting, we don't consider the case where we can observe any number of strong or weak successes but instead the simple case where each type of evidence either exists or does not\footnote{For example, for an interest-$\cI$ observation framework as defined in the previous section, the $\cW$ we are considering now could be the event ``at least one weak study is attempted'', while $\Ws$ is the event ``at least one weak study succeeds'' (and $\Wf$: ``they all fail''); and similarly for the events $\cS$ and $\Ss$ for strong studies.}.

%
%

In this setting, observing weak evidence without observing strong evidence corresponds to the event $\Ws\cap \Ss^C$, where $\Ss^C$ denotes the complement of the event $\Ss$; for the sketch of events $\Ws$ and $\Ss$ shown at left, the right-handside of the figure shows the event $\Ws\cap \Ss^C$.

The sketch in Figure \ref{fig:WSparadox} exhibits three characteristics that are relevant to the improvable evidence paradox:
\begin{itemize}
    \item The events $\cS$ and $\cW$ are positively correlated.  For an interest-$\cI$ observation framework, this is because the probability space includes a random interest level, whose shared value determines the rate at which both strong and weak studies are attempted.
    \item The event $\Ws\cap \cH^C$ is large compared to the event $\Ss\cap \cH^C$; this corresponds to a lower false positive rate for strong evidence.
    \item The event $\Ss\cap \cH$ is not small compared to $\Ws$; i.e., the power of the stronger study is not small compared to the weaker study.
\end{itemize}

Together, these three properties of the sketched example can ensure that the event $\Ws\cap \Ss^C$ (shown at right in Figure \ref{fig:WSparadox}) is strongly negatively correlated with $\cH$. In particular, we can have an improvable evidence paradox, in the sense that
\begin{equation}\label{eq.genparadox}
\Pr(\cH\mid \Ws\cap \Ss^C)<\Pr(\cH\mid \Ws^C\cap \Ss^C);
\end{equation}
in other words, $\cH$ is more likely if we have observed no successful attempts to generate evidence, than if we have observed weak successes but not strong successes.

While it is clear that with freedom to define the events $\cH,$ $\cS,$ $\Ss,$ $\Sf,$ $\cW,$ $\Ws,$ and $\Wf$ as we choose, probability spaces exist satisfying \eqref{eq.genparadox} exist, what Theorems \ref{t.main} and \ref{t.uniform} demonstrate is that this paradoxical nature of improvable evidence can occur in a simple and natural model of evidence generation, even when decisions to seek evidence for a hypothesis are made independent of its truth-value, and attempts to collect strong and weak evidence are correlated merely through a shared rate in independent Poisson distributions.

\smallskip



\section{Discussion and scope}
We view the phenomenon we describe here as widely but not universally applicable to interpretation of weak scientific evidence.  In particular, the most straightforward criticism of scientific evidence is often just that it is very weak, because of a lack of statistical power, moderate effective size, biases in study design, a low prior probability of truth of the hypothesis---or all of the above.  Paradoxes about improvable evidence are not needed to motivate skepticism in weak studies.

Moreover, it would be wrongheaded to use our results as a basis to dismiss any evidence in favor of a hypothesis as long as better evidence is imaginable (as is, indeed, essentially always the case).  Indeed, the finiteness of scientific resources itself limits the extent to which extremely overdispersed distributions (e.g., gamma distributions with very small rates $\beta$) are good priors for $\cI$; in particular, \eqref{eq.main} shows when studies are well-powered and have small false-postive rate---in particular, when the ratio $\gamma/\beta$ is large---we should not expect to face an improvable evidence paradox unless we have extremely uniformative priors for the interest level in hypotheses.  Thus, what our formal results show is that the importance of considering the improvability of scientific evidence grows with the false-positive and false-negative rates of the design of the study being evaluated.  In particular, with many scientific fields operating under conventions where arbitrary fixed cutoffs for Type I error rates act as gatekeepers to publishing, it is when considering borderline cases that understanding the relevance of the improvability of evidence is particularly salient.

Concretely, we propose that to improve the reliability of scientific findings, the presentation and evaluation of research results should evaluate the statistical power and false-positive rate of a study design in the context of the the costs and difficulties that would be required to conduct higher-powered studies.  In particular, if presenting the results of a study with a Type I error rate of $\alpha=.05$, providing a discussion/evaluation of the resource or ethics barriers to study designs with, say, $\alpha=.01$ can provide context for an observer/reader/referee to decide whether they are likely being presented with a study with $\alpha=.05$ because there are reasonable barriers to conducting stronger studies, or simply because stronger studies of the same hypothesis tend not to succeed.

\bigskip

While we prove quantitative conditions under which an improvable evidence paradox occurs, these conditions depend on comparisons between the Type I and Type II error rates of studies and the level of overdispersion in our prior for the interest-level in hypothesis, a prior which is not something that we would necessarily expect observers of scientific evidence to quantify explicitly.  As such, we view the most immediate practical implications of our results as qualitative: namely, when considering evidence in favor of a hypothesis, easily improvable evidence merits special skepticism, unless it already has highly favorable Type I and Type II error rates.

It is important to note that our results here depend on the role of publication bias.  Certainly, one could prove similar paradoxical results in a model with imperfect publication bias, where some negative results leak through.  But if we learn the outcome of every study, not just positive ones, then improvable evidence paradoxes as demonstrated here do not occur, and weak evidence in favor of a hypothesis can be interpreted more robustly.

At the same time, we note that the potentially paradoxical nature of improvable evidence applies not just in situations where better evidence could have been generated by the same investigators or with the same methods, but any time the evidence could be easily and significantly improved with some method and by some investigators, as long as we would expect to learn of such attempts when they are successful.  For example, indirect evidence for a hypothesis should be interpreted with particular caution when direct evidence should be easily obtainable (even if by other methods). 

 In Appendix \ref{s.homogeneous}, we show that even when only studies of one type and strength can be published, an improvable evidence paradox still can arise through the observation of weak $p$-values reported in such studies; i.e., when the evidence provided by the study could be easily improved even by an identical study (on the same sample size, \emph{etc}), simply by achieving a better $p$-value.   This urges particular caution when an excess of ``barely-significant'' studies support a hypothesis, which has been observed in some meta-analyses \cite{ioannidis2017power}.

We note one piece of related work, which showed that when \emph{biases} of studies or investigators are correlated, too-perfect evidence in favor of a hypothesis---i.e., evidence so consistent that it is inconsistent with the false-negative rate that would be expected from unbiased studies---can actually indicate a hypothesis is less likely to be true\footnote{One motivating example from \cite{biasparadox}: If sufficiently many witnesses identify the same suspect in a lineup without a single misidentification, despite the expected unreliability of memory, it can more likely that some corruption or contamination of the process is at play than that we would observe complete consistency across witness recollections, if each truly had an independent chance of failure.} \cite{biasparadox}.  Our framework, on the other hand, does not hinge on correlated biases, and indeed applies even when we assume studies are completely free of bias; we prove our results in a framework in which each attempted study or form of evidence can succeed or fail independently, depending only on the truth of the hypothesis and the relevant Type I/II error rates. Instead, the key form of correlation we assume is just that the likelihood of different types of studies being done on a given hypothesis is correlated, via the level of interest in the hypothesis.

Finally, we note that the psychology literature includes descriptions of a ``weak evidence effect'' \cite{hypotheses,whengoodgoesbad,disputes}, in which presenting study participants with weak evidence in favor of a hypothesis can make them less convinced the hypothesis is true.  Existing explanations of this phenomenon accept it as an erroneous judgment, inconsistent with Bayesian probability, and hypothesize, for example, that weak evidence acts a distraction from the hypothetical possibility of the existence of stronger evidence \cite{hypotheses}, or injects bias into the generation of hypotheses in the human study participant \cite{whengoodgoesbad}.  Our framework, on the other hand, shows that there are situations where weak evidence can have paradoxical effects on the judgement of a purely rational observer reasoning about well-prescribed hypotheses, even in the absence of any bias beyond reporting (publication) bias.

\section{Proofs}
\label{s.proof}
Recall $\eta\sim\Gam(\kappa,\beta)$ (that is, $\eta$ is gamma distributed with shape $\kappa$ and rate $\beta$) if $\eta\geq 0$ and 
\[
\Pr(a\leq \eta\leq b)=\frac{\beta^\kappa}{\Gamma(\kappa)}\int_{x=a}^b x^{\kappa-1}e^{-\beta x}dx \quad\text{for real $0\leq a\leq b$},
\]
where $\Gamma$ is the Gamma function; for example, when $n$ is a positive integer, $\Gamma(n)=(n-1)!$.  Note that the case $\kappa=1$ corresponds to an exponential distribution with rate $\beta$.  Recall also that $\nu\sim \Pos(\lambda)$---that is, $\nu$ is Poisson with mean $\lambda$---if $\nu$ takes only nonnegative integer values and
\[
\Pr(\nu=k)=e^{-\lambda}\frac{\lambda^k}{k!} \quad\text{for $k=0,1,\dots$}.
\]
The following simple fact is convenient in calculations:
\begin{lemma}\label{l.Posflips}
If $\nu$ is distributed as $\Pos(\lambda),$ the number $\zeta$ of successes in $\nu$ independent Bernoulli's with parameter $p$ is distributed as $\zeta\sim \Pos(p \lambda)$.
\end{lemma}
\begin{proof}
We have 
\begin{multline*}
\Pr(\zeta=k)=
\sum_{\ell\geq k}\binom{\ell}{k}p^k(1-p)^{\ell-k}\Pr(\nu=\ell)=
\sum_{\ell\geq k}\binom{\ell}{k}p^k(1-p)^{\ell-k}e^{-\lambda}\frac{\lambda^\ell}{\ell!}\\
=
e^{-\lambda}\frac{(p\lambda)^k}{k!}\sum_{j\geq 0}\frac{(\lambda-p\lambda)^j}{j!}
=e^{-\lambda}\frac{(p\lambda)^k}{k!}e^{\lambda-p\lambda}
=e^{-(p\lambda)}\frac{(p\lambda)^k}{k!},
\end{multline*}
as claimed.
\end{proof}

Now to prove Theorem \ref{t.main}, we are interested in calculating $C_{\cH,\cE}$ from \eqref{eq.CHE} when $\cE$ is an event $\OO_{j,k}$ as defined in \eqref{eq.Ojk}.  To that end, we define probabilities
\begin{equation*}
    Q^\T_{j,k}=\Pr\brac{\OO_{j,k} \mid H\text{ is true}}\quad
    Q^\F_{j,k}=\Pr\brac{\OO_{j,k} \mid H\text{ is false}}
\end{equation*}
With these definitions, we have that
\[
C_{\cH,\OO_{j,k}}=\frac{Q^\T_{j,k}}{Q^\F_{j,k}}.
\]
In particular, we have an improvable evidence paradox; namely, that
    \begin{equation}\label{paradox}
\Pr(H\text{ true}|\OO_{j+1,k})<\Pr(H\text{ true}|\OO_{j,k})
\end{equation}
whenever 
\begin{equation}
\frac{Q^\T_{j+1,k}}{Q^\F_{j+1,k}}< \frac{Q^\T_{j,k}}{Q^\F_{j,k}}.
\end{equation}



To compute these probabilities, we integrate over the value $I$ of the gamma-distributed interest-level $\cI$.  In our model, conditioning on the value $I$, the number of weak and strong experiments that are attempted for the hypothesis are drawn independently from Poisson distributions of mean $c_w I$ and $I$, respectively.  Thus by Lemma \ref{l.Posflips}, for each fixed value of $I$, the number of \emph{successful} strong and weak experiments each are also Poisson-distributed, with means depending on the truth of the hypothesis and the parameters $\alpha_w,$ $\gamma_w,$ $\alpha_S$, and $\gamma_S$ (as well as $I$ and $c_w$).  For example, if the hypothesis is true, then for interest-level $\cI=I$, the number of successful strong experiments is Poisson distributed with mean $\gamma_S\cdot I$, and the number of successful weak experiments is Poisson distributed with mean $c_w\alpha_w I$---and, by assumption, the draws from these distributions are independent. We thus have that
\begin{align}
 Q^\F_{j,k}&=\int_{I=0}^\infty
    \brac{\frac{\beta^\kappa}{\Gamma(\kappa)}I^{\kappa-1}e^{-\beta I}}
    \brac{e^{-\alpha_SI} \frac{(\alpha_S I)^k}{k!}}
    \brac{e^{-\alpha_wc_w I}\frac{(\alpha_wc_wI)^j}{j!}}
dI\label{falseksuccess}\\
Q^\T_{j,k}&=\int_{I=0}^\infty\brac{\frac{\beta^\kappa}{\Gamma(\kappa)}I^{\kappa-1}e^{-\beta I}}
\brac{e^{-\gamma_S I} \frac{(\gamma_S I)^k}{k!}}
\brac{e^{-\gamma_w c_w I} \frac{(\gamma_w c_w I)^j}{j!}}
dI,\label{trueksuccess}
\end{align}
which we rewrite as
\begin{align*}
 Q^\F_{j,k}&=\frac{\beta^\kappa}{\Gamma(\kappa)}\frac{\alpha_S^k(c_w \alpha_w)^j}{j!k!}\int_{I=0}^\infty
    e^{-(\beta+\alpha_S+c_w \alpha_w)I}
    I^{j+k+\kappa-1}
dI\\
Q^\T_{j,k}&=\frac{\beta^\kappa}{\Gamma(\kappa)}\frac{\gamma_S^k(c_w \gamma_w)^j}{j!k!}\int_{I=0}^\infty e^{-(\beta+\gamma_S+c_w \gamma_w)I}I^{j+k+\kappa-1}dI
\end{align*}
In particular, using that 
\begin{equation}\label{eq:gammaint}
\int_{0}^\infty e^{-a x}x^bdx=\frac{\Gamma(b+1)}{a^{b+1}} \quad \text{for real $a>0$ and $b>-1$,}
\end{equation}
we have
\begin{align*}
 Q^\F_{j,k}
&=\frac{\beta^\kappa\Gamma(j+k+\kappa)}{\Gamma(\kappa)\Gamma(j+1)\Gamma(k+1)}\frac{c_w^j\alpha_w^j\alpha_S^k}{(\beta+\alpha_S+c_w\alpha_w)^{j+k+\kappa}}\\
Q^\T_{j,k}
&=\frac{\beta^\kappa\Gamma(j+k+\kappa)}{\Gamma(\kappa)\Gamma(j+1)\Gamma(k+1)}\frac{c_w^j\gamma_w^j\gamma_S^k}{(\beta+\gamma_S+c_w\gamma_w)^{j+k+\kappa}}
\end{align*}

In particular, 
\begin{equation}
\frac{Q^\T_{j,k}}{Q^\F_{j,k}}=
\brac{\frac{\gamma_w}{\alpha_w}}^j
\brac{\frac{\gamma_S}{\alpha_S}}^k
\left(\frac{\beta+\alpha_S+c_w\alpha_w}{\beta+\gamma_S+c_w\gamma_w}\right)^{j+k+\kappa}
\end{equation}
giving an improvable evidence paradox as in \eqref{paradox} whenever 
\[
\alpha_w(\beta+\gamma_S+c_w\gamma_w)>\gamma_w(\beta+\alpha_S+c_w\alpha_w)
\]
and so whenever
\begin{equation}
\beta<\frac{\gamma_S\alpha_w-\gamma_w \alpha_S}{\gamma_w-\alpha_w}.
\end{equation}
This completes the proof of the Theorem.\qed

Separately from the case of when improvable evidence is actually paradoxical, we can explore the general form of its positive predictive value by computing the ratio 
\begin{equation}
C_{\cH,\OO_{\geq 1,0}}=\frac{Q^\F_{\geq 1,0}}{Q^\T_{\geq 1,0}}.
\end{equation}
To this end, note that we have 
\begin{align*}
    Q^\F_{\geq 0,k}
&=\frac{\beta^\kappa\Gamma(k+\kappa)}{\Gamma(\kappa)\Gamma(k+1)}\frac{\alpha_S^k}{(\beta+\alpha_S)^{k+\kappa}}
&&
    Q^\F_{0,k}
&=\frac{\beta^\kappa\Gamma(k+\kappa)}{\Gamma(\kappa)\Gamma(k+1)}\frac{\alpha_S^k}{(\beta+\alpha_S+c_w\alpha_w)^{k+\kappa}}\\
Q^\T_{\geq 0 ,k}
&=\frac{\beta^\kappa\Gamma(k+\kappa)}{\Gamma(\kappa)\Gamma(k+1)}\frac{\gamma_S^k}{(\beta+\gamma_S)^{k+\kappa}}
&&
Q^\T_{0 ,k}
&=\frac{\beta^\kappa\Gamma(k+\kappa)}{\Gamma(\kappa)\Gamma(k+1)}\frac{\gamma_S^k}{(\beta+\gamma_S+c_w\alpha_w)^{k+\kappa}}
\end{align*}
\noindent so that
\begin{align*}
    Q^\F_{\geq 1, 0}
&=\beta^\kappa\brac{\frac{1}{(\beta+\alpha_S)^\kappa}-\frac{1}{(\beta+\alpha_S+c_w\alpha_w)^\kappa}}\\
Q^\T_{\geq 1, 0}
&=\beta^\kappa\brac{\frac{1}{(\beta+\gamma_S)^\kappa}-\frac{1}{(\beta+\gamma_S+c_w\gamma_w)^\kappa}},
\end{align*}
\noindent giving
\begin{align*}
Q^\F_{\geq 1,\geq 0}&=1-\brac{\frac{\beta}{\beta+c_w\alpha_w}}^\kappa=1-\brac{1-\frac{c_w\alpha_w}{\beta+c_w\alpha_w}}^\kappa
\\    
Q^\T_{\geq 1,\geq 0}&=1-\brac{\frac{\beta}{\beta+c_w\gamma_w}}^\kappa=1-\brac{1-\frac{c_w\gamma_w}{\beta+c_w\gamma_w}}^\kappa.
\end{align*}
In particular, we have for $\kappa=1$ that
\begin{equation}
    \frac{Q^\T_{\geq 1,0}}{Q^\F_{\geq 1,0}}
    =\frac{\gamma_w}{\alpha_w}\frac{(\beta+\alpha_S)(\beta+\alpha_S+c_w\alpha_w)}{(\beta+\gamma_S)(\beta+\gamma_S+c_w\gamma_w)}.
\end{equation}
In general, $C_{\cH,\OO_{\geq 1,0}}$ (and thus, the positive predictive value of observing $\OO_{\geq 1,0}$) decreases as $\beta$ decreases, and this effect is particularly critical when the ratio $\gamma_w/\alpha_w$ is not large to begin with.

\bigskip
Theorem \ref{t.uniform} follows easily from a similar calculation to the one above.  Writing $\Gamma_\ell(b,a)$ for the lower incomplete gamma function
\[
\Gamma_\ell(b,a)=\int_{x=0}^a e^{-x}x^b dx,
\]
we have in the setting of Theorem \ref{t.uniform} that 
\[
\frac{Q^\T_{j,k}}{Q^\F_{j,k}}=\brac{\frac{\gamma_w}{\alpha_w}}^j\brac{\frac{\gamma_S}{\alpha_S}}^k\brac{\frac{\alpha_S+c_w\alpha_w}{\gamma_S+c_w\gamma_w}}^{j+k+1}\frac{\Gamma_\ell(j+k+1,C(\gamma_S+c_w\gamma_w))}{\Gamma_\ell(j+k+1,C(\alpha_S+c_w\alpha_w))},
\]
giving that $\Pr(\cH\mid \OO_{j+1,k})<\Pr(\cH\mid \OO_{j,k})$ whenever
\begin{equation}\label{eq.unicondition}
\frac{\gamma_w}{\alpha_w}<\frac{\gamma_S+c_w\gamma_w}{\alpha_S+c_w\alpha_w}\frac{\Gamma_\ell(j+k+2,C(\alpha_S+c_w\alpha_w))\Gamma_\ell(j+k+1,C(\gamma_S+c_w\gamma_w))}{\Gamma_\ell(j+k+1,C(\alpha_S+c_w\alpha_w))\Gamma_\ell(j+k+2,C(\gamma_S+c_w\gamma_w))}.
\end{equation}
This inequality will be satisfied for sufficiently large $C$, since $\frac{\gamma_S}{\alpha_S}>\frac{\gamma_w}{\alpha_w}$ implies that 
\[
\frac{\gamma_S+c_w\gamma_w}{\alpha_S+c_w\alpha_w}>\frac{\gamma_w}{\alpha_w},
\]
and by definition we have that 
\[
\lim_{\lambda\to \infty}\Gamma_\ell(m,\lambda)=\Gamma(m),
\]
ensuring that by making $C$ large (for fixed $j+k$, $\alpha_w,\gamma_w,\alpha_S,\gamma_S$, and $c_w$) we can make the ratio of lower incomplete Gamma functions in \eqref{eq.unicondition} arbitrarily close to 1.
\qed

\subsubsection*{Acknowledgment} We thank Ben Recht for helpful comments on a draft.

\bibliographystyle{plain} 
\bibliography{evidence}

\begin{thebibliography}{10}

\bibitem{amgen}
C~Glenn Begley and Lee~M Ellis.
\newblock Raise standards for preclinical cancer research.
\newblock {\em Nature}, 483(7391):531--533, 2012.

\bibitem{glenjohn}
C.~Glenn Begley and John~P.A. Ioannidis.
\newblock Reproducibility in science.
\newblock {\em Circulation Research}, 116(1):116--126, 2015.

\bibitem{redefine}
Daniel~J Benjamin, James~O Berger, Magnus Johannesson, Brian~A Nosek, E-J
  Wagenmakers, Richard Berk, Kenneth~A Bollen, Bj{\"o}rn Brembs, Lawrence
  Brown, Colin Camerer, et~al.
\newblock Redefine statistical significance.
\newblock {\em Nature human behaviour}, 2(1):6--10, 2018.

\bibitem{button2013power}
Katherine~S Button, John Ioannidis, Claire Mokrysz, Brian~A Nosek, Jonathan
  Flint, Emma~SJ Robinson, and Marcus~R Munaf{\`o}.
\newblock Power failure: why small sample size undermines the reliability of
  neuroscience.
\newblock {\em Nature reviews neuroscience}, 14(5):365--376, 2013.

\bibitem{economics}
Colin~F Camerer, Anna Dreber, Eskil Forsell, Teck-Hua Ho, J{\"u}rgen Huber,
  Magnus Johannesson, Michael Kirchler, Johan Almenberg, Adam Altmejd, Taizan
  Chan, et~al.
\newblock Evaluating replicability of laboratory experiments in economics.
\newblock {\em Science}, 351(6280):1433--1436, 2016.

\bibitem{open2015estimating}
Open~Science Collaboration.
\newblock Estimating the reproducibility of psychological science.
\newblock {\em Science}, 349(6251):aac4716, 2015.

\bibitem{hypotheses}
Ishita Dasgupta, Eric Schulz, and Samuel~J Gershman.
\newblock Where do hypotheses come from?
\newblock {\em Cognitive psychology}, 96:1--25, 2017.

\bibitem{daveyeditorial}
Matthias Egger and George~Davey Smith.
\newblock Misleading meta-analysis: Lessons from "an effective, safe, simple"
  intervention that wasn't.
\newblock {\em BMJ: British Medical Journal}, 310(6982):752--754, 1995.

\bibitem{funnel}
Matthias Egger, George~Davey Smith, Martin Schneider, and Christoph Minder.
\newblock Bias in meta-analysis detected by a simple, graphical test.
\newblock {\em BMJ}, 315(7109):629--634, 1997.

\bibitem{whengoodgoesbad}
Philip~M Fernbach, Adam Darlow, and Steven~A Sloman.
\newblock When good evidence goes bad: The weak evidence effect in judgment and
  decision-making.
\newblock {\em Cognition}, 119(3):459--467, 2011.

\bibitem{biasparadox}
Lachlan~J. Gunn, François Chapeau-Blondeau, Mark~D. McDonnell, Bruce~R. Davis,
  Andrew Allison, and Derek Abbott.
\newblock Too good to be true: when overwhelming evidence fails to convince.
\newblock {\em Proceedings of the Royal Society A: Mathematical, Physical and
  Engineering Sciences}, 472(2187):20150748, 2016.

\bibitem{whymost}
John P.~A. Ioannidis.
\newblock Why most published research findings are false.
\newblock {\em PLOS Medicine}, 2:null, 08 2005.

\bibitem{ioannidis2017power}
John P.~A. Ioannidis, T.~D. Stanley, and Hristos Doucouliagos.
\newblock The power of bias in economics research.
\newblock {\em The Economic Journal}, 127(605):F236--F265, 2017.

\bibitem{manylabs2}
Richard~A Klein, Michelangelo Vianello, Fred Hasselman, Byron~G Adams,
  Reginald~B Adams~Jr, Sinan Alper, Mark Aveyard, Jordan~R Axt, Mayowa~T
  Babalola, {\v{S}}t{\v{e}}p{\'a}n Bahn{\'\i}k, et~al.
\newblock Many labs 2: Investigating variation in replicability across samples
  and settings.
\newblock {\em Advances in Methods and Practices in Psychological Science},
  1(4):443--490, 2018.

\bibitem{disputes}
Craig~RM McKenzie, Susanna~M Lee, and Karen~K Chen.
\newblock When negative evidence increases confidence: Change in belief after
  hearing two sides of a dispute.
\newblock {\em Journal of Behavioral Decision Making}, 15(1):1--18, 2002.

\bibitem{manifesto}
Marcus~R Munaf{\`o}, Brian~A Nosek, Dorothy~VM Bishop, Katherine~S Button,
  Christopher~D Chambers, Nathalie Percie~du Sert, Uri Simonsohn, Eric-Jan
  Wagenmakers, Jennifer~J Ware, and John Ioannidis.
\newblock A manifesto for reproducible science.
\newblock {\em Nature human behaviour}, 1(1):1--9, 2017.

\bibitem{bayer}
Florian Prinz, Thomas Schlange, and Khusru Asadullah.
\newblock Believe it or not: how much can we rely on published data on
  potential drug targets?
\newblock {\em Nature reviews Drug discovery}, 10(9):712--712, 2011.

\bibitem{SterneDavey}
Jonathan~AC Sterne and George~Davey Smith.
\newblock Sifting the evidence—what's wrong with significance tests?
\newblock {\em Physical therapy}, 81(8):1464--1469, 2001.

\end{thebibliography}

\appendix

\section{The homogeneous experiment setting}\label{s.homogeneous}

In this appendix we demonstrate that an improvable evidence paradox can also occur even when we assume that only studies of a single type and power are being conducted of a hypothesis, when the result of such a study is not just success or failure but a $p$-value (which we'll denote by $\rho$) capturing the level of statistical certainty which was achieved in the study.

Roughly, when such a study is capable of producing better $p$-values than were actually obtained, the fact that we observed weaker $p$-values rather than stronger ones can already give rise to the paradox.

In particular, in the model for our result in this section, we assume that:
\begin{itemize}
    \item The possible experimental outcomes are either failure, or success returning a value $\rho\in (0,\alpha]$.
    \item When the hypothesis is false, the probability that $\rho<x$ is given by $a(x)$, for a function satisfying $a(x)\leq x$ (so $\rho$ is a valid $p$-value) and $a(\alpha)=\alpha$.
    \item When the hypothesis is true, the probability of success with $\rho<x$ is given by $\gamma(x)$, for a function $\gamma(x)$ capturing the power of the experiment to achieve a significance level $<x$.
     \end{itemize}
We assume the number of experiments attempted is distributed as $\Pos(\cI)$ for a gamma-distributed interest-level $\cI$ (in particular, the number of experiments follows a negative binomial distribution).

Now for $p<\alpha$ we consider the events 
\[
\OO^{j,p}=\text{Observe exactly $j$ successes, all with $\rho\geq p$}.
\]
We have an improvable evidence paradox in this framework when for all $j\geq 0$,
\[
\Pr(\cH\mid \OO^{j+1,p})<\Pr(\cH\mid \OO^{j,p}).
\]
\begin{theorem}\label{t.homogeneous} In the homogeneous observation framework as defined above, we have an improvable evidence paradox if 
\begin{equation}
\frac{\alpha-p}{\gamma(\alpha)-\gamma(p)}>\frac{\beta+\alpha}{\beta+\gamma(\alpha)},
\end{equation}
or equivalently if
\begin{equation}
    \beta<\frac{\alpha\gamma(p)-p\gamma(\alpha)}{(\gamma(\alpha)-\gamma(p))-(\alpha-p)}.
\end{equation}
In particular, any time
\begin{equation}\label{eq.homocondition}
\frac{\gamma(p)}{p}>\frac{\gamma(\alpha)}{\alpha}>1,
\end{equation}
there is a sufficiently small $\beta>0$ giving rise to an improvable evidence paradox.
\end{theorem}

To understand the applicability of the result, we show \eqref{eq.homocondition} is satisfied for the example of studies based on statistical tests using tails of normal distributions.  In particular, letting 
\[
\Phi(x)=\frac{1}{\sqrt{2\pi}}\int_{-\infty}^xe^{-t^2/2}dt.
\]
be the cumulative distribution function of the normal distribution, we have:
\begin{theorem}\label{t.halfpower}
If we have for some constant $C$ that
\begin{equation}\label{eq.gammaform}
\gamma(x)=\frac{1}{\sqrt{2\pi}}\int_{t=\Phi^{-1}(1-x)-C}^\infty e^{-t^2/2}dt,
\end{equation}
then we have for any $p<\alpha$ that $p/\alpha<\gamma(p)/\gamma(\alpha)$---that is, \eqref{eq.homocondition} is satisfied.
\end{theorem}
We can motivate the form of $\gamma$ in \eqref{eq.gammaform} using the example a one-sided $T$-test  testing for the effect of an intervention on an outcome that is measured before $b_i$ and after $a_i$ the intervention across $n$ samples. The null hypothesis is that the differences $d_i=a_i-b_i$ follow a normal distribution with mean $\mu_d=0$ and the alternative hypothesis 
is that they follow a normal distribution with mean $\mu_d\geq E$. Denoting  $\bar{d}$ and $\sigma_d$ to be the population mean  and standard deviation respectively the test statistic is
\begin{equation}
    T=\frac{\bar{d}}{\sigma_d/\sqrt{n}}
\end{equation}

Under the null hypothesis with $n$ large $T$ is well approximated by the standard normal distribution. We can compute the rejection threshold corresponding to $x$ as $T_x=\Phi^{-1}(1-x)$ for the cumulative distribution function
\[
\Phi(x)=\frac{1}{\sqrt{2\pi}}\int_{-\infty}^xe^{-t^2/2}dt.
\]
If the alternative hypothesis is true we have $\mu_d=E'$ for some $E'\geq E$ and thus that
\begin{align*}
\gamma(x) &=\Pr(T\geq T_x\mid \mu_d=E')\\
&=\Pr\brac{\frac{\bar{d}}{\sigma_d/\sqrt{n}}>T_x\middle| \mu_d=E'}\\
&=\Pr\brac{\frac{\bar{d}-E'}{\sigma_d/\sqrt{n}}>T_x-\frac{E'}{\sigma_d/\sqrt{n}}\middle|\mu_d=E'}\\
&=1-\Pr\brac{\frac{\bar{d}-E'}{\sigma_d/\sqrt{n}}\leq T_x-\frac{E'}{\sigma_d/\sqrt{n}}\middle|\mu_d=E'}
\end{align*}
The quantity $\frac{E'}{(\sigma_d)/\sqrt{n}}$ doesn't depend on $x$, so it is some constant $C$ as in \eqref{eq.gammaform}, and for large $n$, we have that $\frac{\bar{d}-\mu_d}{(\sigma_d)/\sqrt{n}}$ well-approximated by a standard normal distribution; thus $\gamma(x)$ is well approximated by the form in \eqref{eq.gammaform}.

\bigskip

\begin{proof}[Proof of Theorem \ref{t.homogeneous}]
We define the probabilities
\begin{equation*}
    R^\T_{j,p,\alpha}=\Pr\brac{\OO^{j,p} \mid H_i\text{ is true}}\quad
    R^\F_{j,p,\alpha}=\Pr\brac{\OO^{j,p} \mid H_i\text{ is false}}
\end{equation*}
With these definitions, and writing $P_H=\Pr(H_i\text{ true})$, we have have an improvable evidence paradox; namely, that
\begin{equation}\label{eq.homparadox}
\Pr(H_i\text{ true}|\OO^{j+1,p})<\Pr(H_i\text{ true}|\OO^{j,p})
\end{equation}
whenever 
\begin{equation}
\frac{R^\F_{j+1,k}}{R^\T_{j+1,k}}> \frac{R^\F_{j,p,\alpha}}{R^\T_{j,p,\alpha}}.
\end{equation}


We now let $I$ be gamma distributed with shape $\kappa=\mu^2/\sigma^2>0$ and rate $\beta$ and suppose the number of experiments conducted is then independently chosen from a Poisson distribution of mean $I$, where $c_S>0$ is a constant. 

We have
\begin{align}
 R^\F_{j,p,\alpha}&=\int_{I=0}^\infty
    \frac{\beta^\kappa}{\Gamma(\kappa)}I^{\kappa-1}e^{-\beta I}
    e^{-(a(\alpha)-a(p)) I}\frac{(a(\alpha)-a(p))^jI^j}{j!}e^{-a(p) I}
dI\\
R^\T_{j,p,\alpha}&=\int_{I=0}^\infty   \frac{\beta^\kappa}{\Gamma(\kappa)}I^{\kappa-1}e^{-\beta I}
    e^{-(\gamma(\alpha)-\gamma(p)) I}\frac{(\gamma(\alpha)-\gamma(p))^jI^j}{j!}e^{-\gamma(p) I}dI,
\end{align}
which we rewrite as
\begin{align*}
 R^\F_{j,p,\alpha}&=\frac{\beta^\kappa}{\Gamma(\kappa)}\frac{(a(\alpha)-a(p))^j}{j!}\int_{I=0}^\infty
    e^{-(\beta+a(\alpha))I}
    I^{j+\kappa-1}
dI\\
R^\T_{j,p,\alpha}&=\frac{\beta^\kappa}{\Gamma(\kappa)}\frac{(\gamma(\alpha)-\gamma(p))^j}{j!}\int_{I=0}^\infty e^{-(\beta+\gamma(\alpha))I}I^{j+\kappa-1}dI
\end{align*}
In particular, from \eqref{eq:gammaint} we have
\begin{align*}
 R^\F_{j,p,\alpha}
&=\frac{\beta^\kappa\Gamma(j+\kappa)}{\Gamma(\kappa)\Gamma(j+1)}\frac{(a(\alpha)-a(p))^j}{(\beta+a(\alpha))^{j+\kappa}}\\
R^\T_{j,p,\alpha}
&=\frac{\beta^\kappa\Gamma(j+\kappa)}{\Gamma(\kappa)\Gamma(j+1)}\frac{(\gamma(\alpha)-\gamma(p))^j}{(\beta+\gamma(\alpha))^{j+\kappa}}
\end{align*}

In particular, 
\begin{equation}
\frac{R^\F_{j,p,\alpha}}{R^\T_{j,p,\alpha}}=\brac{\frac{a(\alpha)-a(p)}{\gamma(\alpha)-\gamma(p)}}^j\left(\frac{\beta+\gamma(\alpha)}{\beta+a(\alpha)}\right)^{j+\kappa}
\end{equation}
giving an improvable evidence paradox as in \eqref{paradox} whenever 
\[
\frac{a(\alpha)-a(p)}{\gamma(\alpha)-\gamma(p)}>\frac{\beta+a(\alpha)}{\beta+\gamma(\alpha)},
\]
which justifies the condition claimed in the Theorem since $a(\alpha)=\alpha$ and $a(p)\leq p$.
\end{proof}

We conclude by proving Theorem \ref{t.halfpower}.

\begin{proof}[Proof of Theorem \ref{t.halfpower}]
To satisfy \eqref{eq.homocondition}, we are interested in showing that
\begin{equation}\label{eq.intratios}
\frac{\gamma(x)}{x}=\frac{\displaystyle \int_{t=T_x-C}^\infty e^{-t^2/2}dt}
                         {\displaystyle 1-\int_{t=-\infty}^{T_x} e^{-t^2/2}dt}
                   =\frac{\displaystyle \int_{t=T_x-C}^\infty e^{-t^2/2}dt}
                         {\displaystyle \int_{t=T_x}^\infty e^{-t^2/2}dt}
\end{equation}
is a decreasing function of $x$ in the range $x\leq \alpha$.  Writing 
\[
g(y):=\int_{t=y-C}^\infty e^{-t^2/2}dt\quad\text{and}\quad h(y):=\int_{t=y}^\infty e^{-t^2/2}dt,
\]
it suffices to show that $f(\tau)=g(\tau)/h(\tau)$  is an increasing function of $\tau$ in the range $\tau\geq T_\alpha$, since $\tau=T_x$ is a decreasing function of $x$.
Differentiating wih respect to $\tau$, we have that
\begin{equation}
f'(\tau)=\frac{g(\tau)}{h(\tau)}\brac{\frac{g'(\tau)}{g(\tau)}-\frac{h'(\tau)}{h(\tau)}},
\end{equation}
and we thus have that $f'(\tau)>0$ whenever
\begin{equation}\label{eq.gcondition}
     e^{-\tau^2/2}\int_{t=\tau-C}^\infty e^{-t^2/2}dt>e^{-(\tau-C)^2/2}\int_{t=\tau}^\infty e^{-t^2/2}dt.
\end{equation}
Thus it suffices for $h'(\tau)/h(\tau)$ to be decreasing in $\tau$.  Differentiating again, it suffices for this that 
\[
\frac{h''(\tau)}{h'(\tau)}<\frac{h'(\tau)}{h(\tau)}.
\]
We have
\[
h(\tau)=\int_{t=\tau}^\infty e^{-t^2/2}dt\quad\quad h'(\tau)=-e^{-\tau^2/2}\quad\quad h''(\tau)=\tau e^{-\tau^2/2}.
\]
Thus it suffices to have
\[
\tau \int_{t=\tau}^\infty e^{-t^2/2}<e^{-\tau^2/2}.
\]
This is evidently true when $\tau\leq 0$, but also true for all $\tau>0$, since we have
\[
\int_{t=\tau}^\infty e^{-t^2/2}dt<\frac 1 \tau \int_{t=\tau}^\infty te^{-t^2/2}dt=\frac 1 \tau \int_{x=\tau^2/2}^\infty e^{-x}dx=\frac 1 \tau e^{-\tau^2/2}.
\]
\end{proof}

\end{document}